\documentclass[11pt]{article}

\usepackage[top=31mm, bottom=30mm, left=35mm, right=35mm]{geometry}
\usepackage{amsmath,amsthm,amssymb}
\usepackage{cases}

\newenvironment{theorem}[2][Theorem]{\begin{trivlist}
\item[\hskip \labelsep {\bfseries #1}\hskip \labelsep {\bfseries #2.}]}{\end{trivlist}}
\newenvironment{lemma}[2][Lemma]{\begin{trivlist}
\item[\hskip \labelsep {\bfseries #1}\hskip \labelsep {\bfseries #2.}]}{\end{trivlist}}
\newenvironment{definition}[2][Definition]{\begin{trivlist}
\item[\hskip \labelsep {\bfseries #1}\hskip \labelsep {\bfseries #2.}]}{\end{trivlist}}

\newenvironment{remark}[2][Remark]{\begin{trivlist}
\item[\hskip \labelsep {\bfseries #1}\hskip \labelsep {\bfseries #2.}]}{\end{trivlist}}
\newenvironment{acknowledgement}[2][Acknowledgement]{\begin{trivlist}
\item[\hskip \labelsep {\bfseries #1}\hskip \labelsep {\bfseries #2.}]}{\end{trivlist}}
\newenvironment{$proof$}[2][$Proof$]{\begin{trivlist}
\item[\hskip \labelsep {\bfseries #1}\hskip \labelsep {\bfseries #2.}]}{\end{trivlist}}

\begin{document}


\title{On diagonal quantum channels}

\author{Amir R. Arab\\ Moscow Institute of Physics and Technology\\
9 Institutskiy per., Dolgoprudny, Moscow Region,\\ 141701, Russian Federation\\
\texttt{amir.arab@phystech.edu}\\}

\maketitle

\begin{abstract}
In this paper we study diagonal quantum channels and their structure by proving some results and giving most applicable instances of them. Firstly, it is shown that action of every diagonal quantum channel on pure state from computational basis is a convex combination of pure states determined by some transition probabilities. Finally, by using the Cholesky decomposition it is presented an algorithmic method to find an explicit form for Kraus operators of diagonal quantum channels.
\end{abstract}
\vspace{5mm} 
\section{Introduction}
Quantum information theory took shape as a self-consistent and then multidisciplinary area of research from nearly 30 years ago, while its origin can be traced back to the 1950-1960s, which was when the basic ideas of Shannon's information theory were developed. In quantum information theory, the notions of channel and its capacity, giving a measure of ultimate information-processing performance of the channel, play a central role. For a comprehensive introduction to quantum channels, see [1]. A quantum channel is a communication channel which can transmit quantum information, as well as classical information. An example of quantum information is the state of a qubit. Quantum channels are the most general input-output relations which the framework of quantum mechanics allows for arbitrary inputs. Physically, they describe any transmission in space, e.g., through optical fibres, and/or evolution in time, as in quantum memories, from a general open-systems point of view. Mathematically,
they are characterized by linear, completely positive maps acting, in the Schr\"{o}dinger picture, on density operators in a trace-preserving manner.

Diagonal quantum channels have signiﬁcant applications in communication and physics. There are some studies on different types of diagonal channels, for instance depolarizing channels [2-4, 13], transpose depolarizing channels [5] and diagonal channels with constant Frobenius norm (depolarizing, transpose depolarizing, hybrid depolarizing classical, and hybrid transpose depolarizing classical) [6] which are introduced in section 3 of this paper. There are also some works which classify diagonal channels in special cases, for instance see [7] .

 Although we say "Kraus decomposition" it is not his result. It is only a corollary of the famous Stinespring theorem (1955) known e.g. to Choi [8]. Generally, there is no systematic way to find an explicit form for Kraus operators, and they can usually be found by trial and error method. There are some studies on finding the Kraus decomposition, for instance see [9, 10].

 This paper is organized as follows. In section 2, we provide basic definitions and theorems on the Cholesky decomposition and quantum channels. In section 3, we introduce diagonal quantum channels and for instance we give four families of them as the most applicable diagonal channels. In section 4, we prove an interesting property of diagonal channels which says that action of every diagonal channel on pure state from computational basis is a convex combination of pure states. And finally in section 5, by using the Cholesky decomposition we present Kraus representation of diagonal channels.\\
\section{Preliminaries}
\subsection{Cholesky decomposition}
\begin{definition}{1}
Let $\mathbb{C}^{n\times n}$ be the vector space of $n\times n$ matrices over $\mathbb{C}$. A Hermitian matrix $A\in \mathbb{C}^{n\times n}$ is called \emph{positive semi-definite (definite) }, if for every nonzero $x\in \mathbb{C}^n $, $x^{*}Ax\geq 0$ ($x^{*}Ax>0$), where $x^{*}$ is conjugate transpose of $x$.
\end{definition}
The Cholesky decomposition says: Every positive semi-definite matrix $A$ can be decomposed as $A=R^{*}R$, where $R$ is an upper triangular matrix. When $A$ is positive definite, the decomposition is unique [11, p.2]. The Cholesky decomposition goes back to the 1910 as it can be read from his hand-written manuscript deposited by his family in the Archives of the \'{E}cole poly-technique [12, p.91]. This decomposition has various applications such as in the Monte Carlo simulation, non-linear optimization, and linear least squares problems.\\
\subsection{Quantum channels}
\begin{definition}{2}
Linear map $\Phi:\mathbb{C}^{n\times n}\to \mathbb{C}^{l\times l}$ is called \emph{positive}, if for every positive semi-definite matrix $A\in \mathbb{C}^{n\times n}$, $\Phi(A)\in \mathbb{C}^{l\times l}$ is so.
\end{definition}
\begin{definition}{3}
For a positive integer $m$, linear map $\Phi:\mathbb{C}^{n\times n}\to \mathbb{C}^{l\times l}$ is called \emph{m-positive}, if $\Phi\otimes \mathrm{Id}_{m}:\mathbb{C}^{n\times n}\bigotimes\mathbb{C}^{m\times m}\to\mathbb{C}^{l\times l}\bigotimes\mathbb{C}^{m\times m}$ is a positive map, where $\mathrm{Id}_{m}:\mathbb{C}^{m\times m}\to\mathbb{C}^{m\times m}$ is the identity map and $\bigotimes$ is the tensor (Kronecker) product.
\end{definition}
\begin{definition}{4}
Linear map $\Phi:\mathbb{C}^{n\times n}\to \mathbb{C}^{l\times l}$ is called \emph{completely positive}, if for every positive integer m, $\Phi$ is m-positive.
\end{definition}
For example, $\Phi(A) = A^{t}$ ($^{t}$ is transposition) is positive but not completely positive. There are two fundamental theorems which classify completely positive maps:
\begin{theorem}{1} $\Phi:\mathbb{C}^{n\times n}\to \mathbb{C}^{l\times l}$ is completely positive map if and only if there exist $n\times l$ matrices $K_{i}$ such that
\begin{align}
\Phi(A)=\sum_{i}K^{*}_{i}AK_{i},\notag
\end{align}
$K_{i}$'s are called Kraus operators of $\Phi$. Kraus operators are not unique.
\end{theorem}
\begin{theorem}{2}
$\Phi:\mathbb{C}^{n\times n}\to \mathbb{C}^{l\times l}$ is completely positive map if and only if the following matrix is positive semi-definite:
\begin{align}
C_{\Phi}=\sum_{i,j=1}^{n}E_{ij}\bigotimes \Phi(E_{ij}),\notag
\end{align}
where $\{E_{ij}\}_{i,j=1}^{n}$ is the standard basis for $\mathbb{C}^{n\times n}$. $C_{\Phi} $ is called Choi matrix of $\Phi$.
\end{theorem}
To see a proof for these theorems refer [8].
\begin{definition}{5}
Completely positive map $\Phi:\mathbb{C}^{n\times n}\to \mathbb{C}^{l\times l}$ is called \emph{quantum channel}, if it is \emph{trace-preserving}, i.e. $\mathrm{tr}(\Phi(A))=\mathrm{tr}(A)$.
\end{definition}
\begin{definition}{6}
A positive definite matrix $A$ of trace one is called \emph{density matrix}.
\end{definition}
In fact requirement that density matrices are mapped to density matrices leads to the notion of a trace-preserving completely positive map. It is noncommutative analogue of a probability distribution, i.e. a vector whose coordinates are nonnegative and up to one. Trace-preserving property is equivalent to $\sum_{i}K_{i}K^{*}_{i}=I$, where $I$ is $n\times n$ identity matrix.\\

\section{Diagonal channels}
Let us construct an orthonormal basis of Hermitian matrices for $\mathbb{C}^{n\times n}$ with respect to the Hilbert-Schmidt inner product on $\mathbb{C}^{n\times n}$, i.e. $\langle A|B\rangle =\mathrm{tr}(A^{*}B)$.
For $n=2$, the set of Pauli matrices $\{\sigma_{0},\sigma_{1},\sigma_{2},\sigma_{3}\}$ are Hermitian and form an orthogonal basis for $\mathbb{C}^{2\times 2}$, where
\begin{align}
\sigma_{0}=I,\quad \sigma_{1}=\left(
                           \begin{array}{cc}
                             0 & 1 \\
                             1 & 0 \\
                           \end{array}
                         \right) ,\quad \sigma_{2}=\left(
                                                \begin{array}{cc}
                                                  0 & -i \\
                                                  i & 0 \\
                                                \end{array}
                                              \right) ,\quad \sigma_{3}=\left(
                                                                    \begin{array}{cc}
                                                                      1 & 0 \\
                                                                      0 & -1 \\
                                                                    \end{array}
                                                                  \right).\notag
\end{align}
\begin{definition}{7}
For arbitrary $n\geq 3$, the \emph{generalized Pauli matrices} are defined in the following way:
\begin{center}
$\sigma_{0}=I;$
\end{center}
{\tiny\[\sigma_{11}=\left(
              \begin{array}{cccc}
                0 & 1 &... & 0 \\
                1 & 0 &... & 0 \\
               ... &...&... &... \\
                0 & 0 &... & 0 \\
              \end{array}
            \right), \sigma_{12}=\left(
                       \begin{array}{ccccc}
                         0 & 0 & 1 &...& 0 \\
                         0 & 0 & 0 &... & 0 \\
                         1 & 0 & 0 &... & 0 \\
                        ... &... &... &... &... \\
                         0 & 0 & 0 &... & 0 \\
                       \end{array}
                     \right), ..., \sigma_{1N}=\left(
              \begin{array}{cccc}
                0 & 0 &... & 0 \\
                ... &... &... &... \\
                0 &...& 0 & 1 \\
                0 &... & 1 & 0 \\
              \end{array}
            \right); \]}
{\tiny\[\sigma_{21}=\left(
              \begin{array}{cccc}
                0 & -i &... & 0 \\
                i & 0 &... & 0 \\
               ... &...&... &... \\
                0 & 0 &... & 0 \\
              \end{array}
            \right), \sigma_{22}=\left(
                       \begin{array}{ccccc}
                         0 & 0 & -i &...& 0 \\
                         0 & 0 & 0 &... & 0 \\
                         i & 0 & 0 &... & 0 \\
                        ... &... &... &... &... \\
                         0 & 0 & 0 &... & 0 \\
                       \end{array}
                     \right), ..., \sigma_{2N}=\left(
              \begin{array}{cccc}
                0 & 0 &... & 0 \\
                ... &... &... &... \\
                0 &...& 0 & -i \\
                0 &... & i & 0 \\
              \end{array}
            \right); \]}
{\tiny\[\sigma_{31}=\left(
              \begin{array}{cccc}
                1 & 0 &... & 0 \\
                0 & -1 &... & 0 \\
               ... &...&... &... \\
                0 & 0 &... & 0 \\
              \end{array}
            \right), \sigma_{32}=\left(
                       \begin{array}{ccccc}
                         1 & 0 & 0 &...& 0 \\
                         0 & 0 & 0 &... & 0 \\
                         0 & 0 & -1 &... & 0 \\
                        ... &... &... &... &... \\
                         0 & 0 & 0 &... & 0 \\
                       \end{array}
                     \right), ..., \sigma_{3N}=\left(
              \begin{array}{cccc}
                0 & 0 &... & 0 \\
                ... &... &... &... \\
                0 &...& 1 & 0 \\
                0 &... & 0 & -1 \\
              \end{array}
            \right); \]}
where $N=\frac{n(n-1)}{2}$.
\end{definition}
The generalized Pauli matrices are Hermitian and orthogonal, but they couldn't form a basis for $\mathbb{C}^{n\times n}$. We remove matrices $\sigma_{31}, \sigma_{32},...,\sigma_{3N}$ and add the following matrices:
{\tiny\[A_{1}=\left(
              \begin{array}{cccc}
                1 & 0 &... & 0 \\
                0 & -1 &... & 0 \\
               ... &...&... &... \\
                0 & 0 &... & 0 \\
              \end{array}
            \right), A_{2}=\left(
                       \begin{array}{ccccc}
                         1 & 0 & 0 &...& 0 \\
                         0 & 1 & 0 &... & 0 \\
                         0 & 0 & -2 &... & 0 \\
                        ... &... &... &... &... \\
                         0 & 0 & 0 &... & 0 \\
                       \end{array}
                     \right), ..., A_{n-1}=\left(
              \begin{array}{cccc}
                1 &... & 0 & 0 \\
                ... &... &... &... \\
                0 &...& 1 & 0 \\
                0 &... & 0 & -(n-1) \\
              \end{array}
            \right). \]}
We consider
\begin{align}
\textbf{$\beta$}=\Bigg\{\frac{\sigma_{0}}{\sqrt{n}},\frac{\sigma_{11}}{\sqrt{2}},...,\frac{\sigma_{1N}}{\sqrt{2}},\frac{\sigma_{21}}{\sqrt{2}},...,\frac{\sigma_{2N}}{\sqrt{2}},\frac{A_{1}}{\sqrt{2}},\frac{A_{2}}{\sqrt{6}},...,\frac{A_{n-1}}{\sqrt{(n-1)n}}\Bigg\},\notag
\end{align}
$\beta$ is an orthonormal basis of Hermitian matrices [6].
\begin{definition}{8}
Quantum channel $\Phi:\mathbb{C}^{n\times n}\to \mathbb{C}^{n\times n}$ is called \emph{diagonal}, if its representation with respect to the basis $\beta$ is diagonal, i.e.
\begin{align}
\Phi=\mathrm{diag}(1,a_{1},a_{2},...,a_{n^{2}-1}).\notag
\end{align}
\end{definition}
Here we introduce four families of diagonal channels:\\
1. $\Phi=\mathrm{diag}(1,\underbrace{p,...,p}_{N},\underbrace{p,...,p}_{N},\underbrace{p,...,p}_{n-1}),$ where $-\frac{1}{n^{2}-1}\leq p\leq 1$,\\
2. $\Phi=\mathrm{diag}(1,\underbrace{p,...,p}_{N},\underbrace{-p,...,-p}_{N},\underbrace{p,...,p}_{n-1}),$ where $-\frac{1}{n-1}\leq p\leq \frac{1}{n+1}$,\\
3. $\Phi=\mathrm{diag}(1,\underbrace{-p,...,-p}_{N},\underbrace{-p,...,-p}_{N},\underbrace{p,...,p}_{n-1}),$ where $-\frac{1}{2n-1}\leq p\leq \frac{1}{(n-1)^{2}}$,\\
4. $\Phi=\mathrm{diag}(1,\underbrace{-p,...,-p}_{N},\underbrace{p,...,p}_{N},\underbrace{p,...,p}_{n-1}),$ where $-\frac{1}{n-1}\leq p\leq \frac{1}{n+1}$.\\
Quantum channels 1, 2, 3, and 4 are called depolarizing, transpose depolarizing, hybrid depolarizing classical, and hybrid transpose depolarizing classical respectively.
They are among most widely used channels in science and technology.
\section{Transition probabilities}
In this section we prove an interesting property of diagonal channels which is formulated in the following theorem.
\begin{theorem}{3}
\emph{For every diagonal quantum channel} $\Phi$, \emph{there is a collection of transition probabilities} $\{P_{kj}\}_{j=1}^{n}$
$,i.e.\hspace{1mm} P_{kj}\geq 0, \sum_{j=1}^{n}P_{kj}=1$ \emph{such that}
\begin{align}
\Phi(|k\rangle\langle k|)=\sum_{j=1}^{n}P_{kj}|j\rangle\langle j| \qquad (k=1,2,...,n).\notag
\end{align}

\end{theorem}
\begin{proof}
For convenience we write basis $\beta$ in the following form:
\begin{align}
 \beta=\{e_{0},e_{11},...,e_{1N},e_{21},...,e_{2N},e_{31},...,e_{3(n-1)}\}.\notag
\end{align}
Let $\Phi$ be a diagonal channel and its matrix representation with respect to the basis $\beta$ is
\begin{align}
\Phi=\mathrm{diag}(1,r_{1},r_{2},...,r_{N},s_{1},s_{2},...,s_{N},t_{1},t_{2},...,t_{n-1}).\notag
\end{align}
Suppose that $\{|j\rangle\}$ is the standard basis for $\mathbb{C}^{n}$. We can write
\begin{align}
E_{kk}=|k\rangle\langle k|=c_{0}e_{0}+\sum_{p=1}^{2}\sum_{q=1}^{N}c_{pq}e_{{pq}}+\sum_{q=1}^{n-1}c_{3q}e_{3q},\notag
\end{align}
and calculate the coefficients in the following way:
\begin{gather*}
c_{0}=\langle e_{0}|E_{kk}\rangle=\frac{1}{\sqrt{n}},\quad c_{pq}=\langle e_{pq}|E_{kk}\rangle=0 ,\quad 1\leq p\leq 2,\quad 1\leq q\leq N,\\
c_{3q}=\langle e_{3q}|E_{kk}\rangle=0,\quad 1\leq q\leq k-2,\quad c_{3(k-1)}=\langle e_{3(k-1)}|E_{kk}\rangle=-\sqrt{\frac{k-1}{k}},\\
c_{3q}=\langle e_{3q}|E_{kk}\rangle=\frac{1}{\sqrt{q(q+1)}},\quad k\leq q\leq n-1.\notag
\end{gather*}

\flushleft{Therefore for $k=1,2,...,n-1$}
\begin{gather*}
E_{kk}=\frac{1}{\sqrt{n}}e_{0}-\sqrt{\frac{k-1}{k}}e_{3(k-1)}+\sum_{i=k}^{n-1}\frac{1}{\sqrt{i(i+1)}}e_{3i},\\
E_{nn}=\frac{1}{\sqrt{n}}e_{0}-\sqrt{\frac{n-1}{n}}e_{3(n-1)},\notag
\end{gather*}
and applying $\Phi$ we obtain
\begin{gather}
\Phi(E_{kk})=\frac{1}{\sqrt{n}}e_{0}-t_{k-1}\sqrt{\frac{k-1}{k}}e_{3(k-1)}+\sum_{i=k}^{n-1}\frac{t_{i}}{\sqrt{i(i+1)}}e_{3i},\\
\Phi(E_{nn})=\frac{1}{\sqrt{n}}e_{0}-t_{n-1}\sqrt{\frac{n-1}{n}}e_{3(n-1)}.
\end{gather}
It follows from (1)
\begin{numcases}{[\Phi(E_{kk})]_{jj}=}
\frac{1}{n}-\frac{t_{k-1}}{k}+\sum_{i=k}^{n-1}\frac{t_{i}}{i(i+1)} & $\qquad 1\leq j\leq k-1$\label{positive},\\
\frac{1}{n}+\frac{(k-1)t_{k-1}}{k}+\sum_{i=k}^{n-1}\frac{t_{i}}{i(i+1)} & $\qquad j=k$\label{negative},\\
\frac{1}{n}-\frac{t_{j-1}}{j}+\sum_{i=j}^{n-1}\frac{t_{i}}{i(i+1)} & $\qquad k+1\leq j\leq n$\label{positive}.
\end{numcases}
Taking into account (2) we get
\begin{numcases}{[\Phi(E_{nn})]_{jj}=}
\frac{1}{n}-\frac{t_{n-1}}{n} & $\qquad 1\leq j\leq n-1$\label{positive},\\
\frac{1}{n}+\frac{(n-1)t_{n-1}}{n} & $\qquad j=n$\label{negative}.
\end{numcases}
Now we consider the Choi matrix of $\Phi$. We note that $\Phi(E_{kk}),(k=1,2,...,n)$ are $n\times n$ diagonal matrices and their diagonals are given by (3)-(7). By theorem 2, $C_{\Phi}$ is positive semi-definite, then all of its diagonal entries are nonnegative. For $1\leq k,j\leq n$, we define $P_{kj}$ in the following manner:
\begin{align}
P_{kj}=[\Phi(E_{kk})]_{jj}.\notag
\end{align}
Hence, $P_{kj}\geq 0$ and for every $k\in \{1,2,...,n\}$, the equality $\sum_{j=1}^{n}P_{kj}=1$ holds true. Therefore $\{P_{kj}\}$ are the desired transition probabilities.
\end{proof}
\begin{remark}{1}
Equality $\sum_{j=1}^{n}P_{kj}=1$ can be derived from the fact that $\Phi$ is trace-preserving.
\end{remark}
\section{Kraus representation for diagonal channel}
Before we formulate the result of this section, we need to prove the following two lemmas.
\begin{lemma}{1}
\emph{Let} $\kappa=(x_{1},x_{2},...,x_{n})$ \emph{where} $x_{i}$\emph{'s are} \emph{rows of} $n\times n$ \emph{matrix} $K$, \emph{then} $(K^{*}E_{ij}K)_{1\leq i,j\leq n}=(x_{i}^{*}x_{j})_{1\leq i,j\leq n}=\kappa^{*}\kappa$.
\end{lemma}
\begin{proof}
For $m=1,2,...,n$, let $x_{m}=\left(
             \begin{array}{cccccc}
               x_{m1} & x_{m2} & . & . & . & x_{mn} \\
             \end{array}
           \right)$, then
\begin{gather*}
K^{*}E_{ij}K=\left(
               \begin{array}{cccccc}
                 x_{1}^{*} & x_{2}^{*} & . & . & . & x_{n}^{*}
                 \end{array}
             \right)E_{ij}\left(
                                                                                                               \begin{array}{c}
                                                                                                                 x_{1} \\
                                                                                                                 x_{2} \\
                                                                                                                 . \\
                                                                                                                 . \\
                                                                                                                 . \\
                                                                                                                 x_{n} \\
                                                                                                               \end{array}
                                                                                                             \right)\\
=(\sum_{s=1}^{n}\overline{x_{is}}E_{si})E_{ij}(\sum_{t=1}^{n}x_{jt}E_{jt})=\sum_{s,t=1}^{n}\overline{x_{is}}x_{jt}E_{st}=x_{i}^{*}x_{j}.
\end{gather*}
On the other hand, $\kappa^{*}\kappa=(x_{i}^{*}x_{j})_{1\leq i,j\leq n}$.
\end{proof}
\begin{lemma}{2}
\emph{Let} $\Phi:\mathbb{C}^{n\times n}\to \mathbb{C}^{n\times n}$ \emph{be a quantum channel,} $C_{\Phi}$ \emph{be its Choi matrix, and} $C_{\Phi}=R^{*}R$ \emph{for some matrix $R$}. \emph{If} $\kappa_{i}$\emph{'s} \emph{are rows of} $R$, \emph{and} $K_{i}$\emph{'}\emph{s are associated matrices to} $\kappa_{i}$\emph{'s in lemma} 1 $(1\leq i\leq n^{2})$ \emph{then} $\{K_{i}\}_{i=1}^{n^{2}}$ \emph{is a set of Kraus operators of} $\Phi$.
\end{lemma}
\begin{proof}
Choi matrix of $\Phi$ is in the following form:
\begin{gather}
C_{\Phi}=\sum_{i,j=1}^{n}E_{ij}\bigotimes \Phi(E_{ij})=(\Phi(E_{ij}))_{1\leq i,j\leq n}.\notag
\end{gather}
On the other hand,
\begin{align}
C_{\Phi}=R^{*}R=\left(
                  \begin{array}{cccccc}
                    \kappa_{1}^{*} & \kappa_{2}^{*} & . & . & . & \kappa_{n^{2}}^{*} \\
                  \end{array}
                \right)\left(
                         \begin{array}{c}
                           \kappa_{1} \\
                           \kappa_{2} \\
                           . \\
                           . \\
                           . \\
                           \kappa_{n^{2}} \\
                         \end{array}
                       \right)=\sum_{i=1}^{n^{2}}\kappa_{i}^{*}\kappa_{i}.\notag
\end{align}
Then applying lemma 1, $(\Phi(E_{ij}))_{1\leq i,j\leq n}=\sum_{l=1}^{n^{2}}(K_{l}^{*}E_{ij}K_{l})_{1\leq i,j\leq n}$. Therefore for any $A$, $\Phi(A)=\sum_{i=1}^{n^{2}}K_{i}^{*}AK_{i}$.
\end{proof}
Now we are in a position to assert the main result of this section. We are going to find an explicit form for Kraus operators of channels 1-4 by using the Cholesky decomposition. Here we solve this problem for channel 3 and for other channels it is solved similarly.
\begin{theorem}{4}
\emph{For hybrid depolarizing classical quantum channel}
\begin{align*}
\Phi=\mathrm{diag}(1,\underbrace{-p,...,-p}_{N},\underbrace{-p,...,-p}_{N},\underbrace{p,...,p}_{n-1}),
\end{align*}
\emph{Kraus operators can be determined in the following explicit form:}
{\[K_{1}=\left(
        \begin{array}{cccc}
          \sqrt{a_{0}} & 0 & ... & 0 \\
          0 & \frac{b_{0}}{\sqrt{a_{0}}} & ... & 0 \\
          ... & ... & ... & ... \\
          0 & 0 & ... & \frac{b_{0}}{\sqrt{a_{0}}} \\
        \end{array}
      \right),\quad K_{2}=\left(
                      \begin{array}{cccc}
                        0 & \sqrt{\frac{1-p}{n}} & ... & 0 \\
                        0 & 0 & ... & 0 \\
                        ... & ...& ... & ... \\
                        0 & 0 & ... & 0 \\
                      \end{array}
                    \right), \]}\\
{\[...,K_{n}=\left(
        \begin{array}{cccc}
         0 & 0 & ... & \sqrt{\frac{1-p}{n}} \\
         0 & 0 & ... & 0 \\
         ... & ... & ... & ... \\
         0 & 0 & ... & 0 \\
          \end{array}
         \right),\quad K_{n+1}=\left(
          \begin{array}{cccc}
            0 & 0 & ... & 0 \\
            \sqrt{\frac{1-p}{n}} & 0 & ... & 0 \\
            ... & ... & ... & ... \\
            0 & 0 & ... & 0 \\
          \end{array}
        \right),\]}\\
{\[K_{n+2}=\left(
                          \begin{array}{cccc}
                            0 & 0 & ... & 0 \\
                            0 & \sqrt{a_{1}} & ... & 0 \\
                            ... & ... & ... & ... \\
                            0 & 0 & ... & \frac{b_{1}}{\sqrt{a_{1}}} \\
                          \end{array}
                        \right),...,K_{2n}=\left(
                                             \begin{array}{cccc}
                                               0 & 0 & ... & 0 \\
                                               0 & 0 & ... & \sqrt{\frac{1-p}{n}} \\
                                               ... & ... & ... & ... \\
                                               0 & 0 & ... & 0 \\
                                             \end{array}
                                           \right), \]}\\
{\[...,K_{n^{2}-1}=\left(
                                            \begin{array}{cccc}
                                              0 & ... & 0 & 0 \\
                                              0 & ... & 0 & 0 \\
                                              ... & ... & ... & ... \\
                                              0 & ... & \sqrt{\frac{1-p}{n}} & 0 \\
                                            \end{array}
                                          \right),\quad K_{n^{2}}=\left(
                                                              \begin{array}{cccc}
                                                                0 & 0 & ... & 0 \\
                                                                0 & 0 & ... & 0 \\
                                                                ... & ... & ... & ... \\
                                                                0 & 0 & ... & \sqrt{a_{n-1}} \\
                                                              \end{array}
                                                            \right),\]}
\emph{where} $a_{m}=(2p+\frac{1-p}{n})(1+\frac{-p}{-pm+2p+\frac{1-p}{n}})$ \emph{for} $m=1,2,...,n-1$; $b_{m}=(2p+\frac{1-p}{n})(\frac{-p}{-pm+2p+\frac{1-p}{n}})$ \emph{for} $m=1,2,...,n-2$; $a_{0}=p+\frac{1-p}{n}$\emph{, and} $b_{0}=-p$.
\end{theorem}
\begin{remark}{2}
The statement and proof of theorem 4 is based on the condition:\\ $(p+\frac{1-p}{n})(\frac{1-p}{n})\neq 0$ (in fact the entries of main diagonal of Choi matrix be positive). In case at least one of them is 0, matrix $R$ in Cholesky decomposition of $C_{\Phi}$ is not necessarily unique and then Kraus operators which are obtained in this way are so.
\end{remark}
\begin{proof}
It follows from (1) that for $t_{1}=t_{2}=...=t_{n-1}=p$,
\begin{align}
\Phi(E_{11})=\frac{1}{\sqrt{n}}e_{0}+\sum_{i=1}^{n-1}\frac{p}{\sqrt{i(i+1)}}e_{3i}=\left(
                                                                                     \begin{array}{cccc}
                                                                                       \alpha & 0 & ... & 0 \\
                                                                                       0 & \beta & ... & 0 \\
                                                                                       ... & ... & ... & ... \\
                                                                                       0 & 0 & ... & \beta \\
                                                                                     \end{array}
                                                                                   \right)
,\notag
\end{align}
where $\alpha=p+\frac{1-p}{n}$ , $\beta=\frac{1-p}{n}$. Similarly we can obtain $\Phi (E_{kk})$ for $k=2, ..., n$.
We note that $E_{12}=\frac{1}{\sqrt{2}}e_{11}+\frac{i}{\sqrt{2}}e_{21}$, hence
\begin{align}
\Phi(E_{12})=\frac{-p}{\sqrt{2}}e_{11}+\frac{-pi}{\sqrt{2}}e_{21}=\left(
                                                                 \begin{array}{cccc}
                                                                   0 & \gamma & ... & 0 \\
                                                                   0 & 0 & ... & 0 \\
                                                                   ... & ... & ... & ... \\
                                                                   0 & 0 & ... & 0 \\
                                                                 \end{array}
                                                               \right),\notag
\end{align}
where $\gamma=-p$. Similarly we can obtain $\Phi(E_{kj})$ for $1\leq k\neq j\leq n$.

Therefore Choi matrix of $\Phi$ has the following form:

\[C_{\Phi}=
  \left(\begin{array}{@{}cccc|cccc|c|cccc@{}}
             \alpha & 0 & ... & 0 & 0 & \gamma & ... & 0 & ... & 0 & 0 & ... & \gamma\\
             0 & \beta & ... & 0 & 0 & 0 & ... & 0 & ... & 0 & 0 & ... & 0 \\
             ... & ... & ... & ... & ... & ... & ... & ... & ... & ... & ... & ... & ... \\
             0 & 0 & ... & \beta & 0 & 0 & ... & 0 & ... & 0 & 0 & ... & 0 \\\hline
             0 & 0 & ... & 0 & \beta & 0 & ... & 0 & ... & 0 & 0 & ... & 0 \\
             \gamma & 0 & ... & 0 & 0 & \alpha & ... & 0 & ... & 0 & 0 & ... & \gamma \\
             ... & ... & ... & ... & ... & ... & ... & ... & ... & ... & ... & ... & ... \\
             0 & 0 & ... & 0 & 0 & 0 & ... & \beta & ... & 0 & 0 & ... & 0 \\\hline
             ... & ... & ... & ... & ... & ... & ... & ... & ... & ... & ... & ... & ... \\\hline
             0 & 0 & ... & 0 & 0 & 0 & ... & 0 & ... & \beta & 0 & ... & 0 \\
             0 & 0 & ... & 0 & 0 & 0 & ... & 0 & ... & ... & ... & ... & ... \\
             ... & ... & ... & ... & ... & ... & ... & ... & ... & 0 & ... & \beta & 0 \\
             \gamma & 0 & ... & 0 & 0 & \gamma & ... & 0 & ... & 0 & ... & 0 & \alpha \\
  \end{array}\right),
\]\\
Let $C_{\Phi}=R^{t}R$ be Cholesky decomposition of $C_{\Phi}$. We divide matrices into 4 blocks as the following:
\[C_{\Phi}=
  \left(\begin{array}{@{}c|c @{}}
             \alpha & A_{10} \\\hline
             A_{10}^{t} & A_{20} \\
  \end{array}\right)=
  \left(\begin{array}{@{}c|c @{}}
             r_{0} & \textbf{0} \\\hline
             R_{10}^{t} & R_{20}^{t} \\
  \end{array}\right)
  \left(\begin{array}{@{}c|c @{}}
             r_{0} & R_{10} \\\hline
             \textbf{0} & R_{20} \\
  \end{array}\right),\]
where $r_{0}=\sqrt{\alpha}$, $R_{10}=\frac{1}{r_{0}}A_{10}$, $A_{20}-R_{10}^{t}R_{10}=R_{20}^{t}R_{20}$, and $R_{20}$ is an upper triangular matrix. Then
\begin{align}
R_{10}=\frac{1}{\sqrt{\alpha}}\left(
                               \begin{array}{cccccccccccc}
                                 0 & ... & 0 & \gamma & 0 & ... & 0 & \gamma & 0 & ... & 0 & \gamma \\
                               \end{array}
                             \right),\notag
\end{align}
we let
\begin{align}
\kappa_{1}=\left(
            \begin{array}{ccccccccccccc}
              \sqrt{\alpha} & 0 & ... & 0 & \frac{\gamma}{\sqrt{\alpha}} & 0 & ... & 0 & \frac{\gamma}{\sqrt{\alpha}} & 0 & ... & 0 & \frac{\gamma}{\sqrt{\alpha}} \\
               \end{array}
                \right),\notag
\end{align}
and its associated $K_{1}$ in lemma 2 as the following:
\begin{align}
K_{1}=\left(
        \begin{array}{cccc}
          \sqrt{\alpha} & 0 & ... & 0 \\
          0 & \frac{\gamma}{\sqrt{\alpha}} & ... & 0 \\
          ... & ... & ... & ... \\
          0 & 0 & ... & \frac{\gamma}{\sqrt{\alpha}} \\
        \end{array}
      \right).\notag
\end{align}
Now we consider $A_{20}-R_{10}^{t}R_{10}$:
\[A_{20}-R_{10}^{t}R_{10}=
  \left(\begin{array}{@{}cccc|cccc|c|cccc@{}}
             \beta & 0 & ... & 0 & 0 & 0 & ... & 0 & ... & 0 & 0 & ... & 0\\
             0 & \beta & ... & 0 & 0 & 0 & ... & 0 & ... & 0 & 0 & ... & 0 \\
             ... & ... & ... & ... & ... & ... & ... & ... & ... & ... & ... & ... & ... \\
             0 & 0 & ... & \beta & 0 & 0 & ... & 0 & ... & 0 & 0 & ... & 0 \\\hline
             0 & 0 & ... & 0 & \beta & 0 & ... & 0 & ... & 0 & 0 & ... & 0 \\
             0 & 0 & ... & 0 & 0 & \alpha-\frac{\gamma^{2}}{\alpha} & ... & 0 & ... & 0 & 0 & ... & \gamma-\frac{\gamma^{2}}{\alpha} \\
             ... & ... & ... & ... & ... & ... & ... & ... & ... & ... & ... & ... & ... \\
             0 & 0 & ... & 0 & 0 & 0 & ... & \beta & ... & 0 & 0 & ... & 0 \\\hline
             ... & ... & ... & ... & ... & ... & ... & ... & ... & ... & ... & ... & ... \\\hline
             0 & 0 & ... & 0 & 0 & 0 & ... & 0 & ... & \beta & 0 & ... & 0 \\
             0 & 0 & ... & 0 & 0 & 0 & ... & 0 & ... & ... & ... & ... & ... \\
             ... & ... & ... & ... & ... & ... & ... & ... & ... & 0 & ... & \beta & 0 \\
             0 & 0 & ... & 0 & 0 & \gamma-\frac{\gamma^{2}}{\alpha} & ... & 0 & ... & 0 & ... & 0 & \alpha-\frac{\gamma^{2}}{\alpha} \\
  \end{array}\right),
\]\\
we note that by applying algorithm $\alpha$ , $\gamma$ are replaced with $\alpha-\frac{\gamma^{2}}{\alpha}$ , $\gamma-\frac{\gamma^{2}}{\alpha}$ respectively.
\[A_{20}-R_{10}^{t}R_{10}=
  \left(\begin{array}{@{}c|c @{}}
             \beta & A_{11} \\\hline
             A_{11}^{t} & A_{21} \\
  \end{array}\right)=
  \left(\begin{array}{@{}c|c @{}}
             r_{1} & \textbf{0} \\\hline
             R_{11}^{t} & R_{21}^{t} \\
  \end{array}\right)
  \left(\begin{array}{@{}c|c @{}}
             r_{1} & R_{11} \\\hline
             \textbf{0} & R_{21} \\
  \end{array}\right),\]
where $R_{20}=\left(\begin{array}{@{}c|c @{}}
             r_{1} & R_{11} \\\hline
             \textbf{0} & R_{21} \\
  \end{array}\right)$, $r_{1}=\sqrt{\beta}$, $R_{11}=\textbf{0}_{1\times(n^{2}-2)}$, and $R_{21}$ is an upper triangular matrix.
  We let $\kappa_{2}=\left(
   \begin{array}{ccccc}
        0 & \sqrt{\beta} & 0 & ... & 0 \\
           \end{array}
            \right)$ and its associated $K_{2}$ in lemma 2 as the following:
\begin{align}
K_{2}=\left(
        \begin{array}{cccc}
          0 & \sqrt{\beta} & ... & 0 \\
          0 & 0 & ... & 0 \\
          ... & ... & ... & ... \\
          0 & 0 & ... & 0 \\
        \end{array}
      \right).\notag
\end{align}
By repeating algorithm, similarly we obtain
{\tiny\[K_{3}=\left(
        \begin{array}{ccccc}
          0 & 0 & \sqrt{\beta} &... & 0 \\
          0 & 0 & 0 & ... & 0 \\
          ... & ... & ... & ... & ... \\
          0 & 0 & 0 & ... & 0 \\
        \end{array}
      \right),...,K_{n}=\left(
                      \begin{array}{cccc}
                        0 & 0 & ... & \sqrt{\beta} \\
                        0 & 0 & ... & 0 \\
                        ... & ...& ... & ... \\
                        0 & 0 & ... & 0 \\
                      \end{array}
                    \right),K_{n+1}=\left(
                                        \begin{array}{cccc}
                                          0 & 0 & ... & 0 \\
                                          \sqrt{\beta} & 0 & ... & 0 \\
                                          ... & ... & ... & ... \\
                                          0 & 0 & ... & 0 \\
                                        \end{array}
                                      \right). \]}\\
Now we reach the following matrix:
\[
  \left(\begin{array}{@{}cccc|ccccc|c|cccc@{}}
             \alpha-\frac{\gamma^{2}}{\alpha} & 0 & ... & 0 & 0 & 0 & \gamma-\frac{\gamma^{2}}{\alpha} & ... & 0 & ... & 0 & 0 & ... & \gamma-\frac{\gamma^{2}}{\alpha}\\
             0 & \beta & ... & 0 & 0 & 0 & 0 & ... & 0 & ... & 0 & 0 & ... & 0 \\
             ... & ... & ... & ... & ... & ... & ... & ... & ... & ... & ... & ... & ... & ... \\
             0 & 0 & ... & \beta & 0 & 0 & 0 & ... & 0 & ... & 0 & 0 & ... & 0 \\\hline
             0 & 0 & ... & 0 & \beta & 0 & 0 & ... & 0 & ... & 0 & 0 & ... & 0 \\
             0 & 0 & ... & 0 & 0 & \beta & 0 & ... & 0 & ... & 0 & 0 & ... & 0 \\
             \gamma-\frac{\gamma^{2}}{\alpha} & 0 & ... & 0 & 0 & 0 & \alpha-\frac{\gamma^{2}}{\alpha} & ... & 0 & ... & 0 & 0 & ... & \gamma-\frac{\gamma^{2}}{\alpha} \\
             ... & ... & ... & ... & ... & ... & ... & ... & ... & ... & ... & ... & ... & ... \\
             0 & 0 & ... & 0 & 0 & 0 & 0 & ... & \beta & ... & 0 & 0 & ... & 0 \\\hline
             ... & ... & ... & ... & ... & ... & ... & ... & ... & ... & ... & ... & ... & ... \\\hline
             0 & 0 & ... & 0 & 0 & 0 & 0 & ... & 0 & ... & \beta & 0 & ... & 0 \\
             0 & 0 & ... & 0 & 0 & 0 & 0 & ... & 0 & ... & ... & ... & ... & ... \\
             ... & ... & ... & ... & ... & ... & ... & ... & ... & ... & 0 & ... & \beta & 0 \\
             \gamma-\frac{\gamma^{2}}{\alpha} & 0 & ... & 0 & 0 & 0 & \gamma-\frac{\gamma^{2}}{\alpha} & ... & 0 & ... & 0 & ... & 0 & \alpha-\frac{\gamma^{2}}{\alpha} \\
  \end{array}\right),
\]\\
by applying algorithm, we obtain
\[K_{n+2}=\left(
        \begin{array}{ccccc}
          0 & 0 & 0 &... & 0 \\
          0 & \sqrt{\alpha-\frac{\gamma^{2}}{\alpha}} & 0 & ... & 0 \\
          0 & 0 & \frac{\gamma-\frac{\gamma^{2}}{\alpha}}{\sqrt{\alpha-\frac{\gamma^{2}}{\alpha}}} &... & 0 \\
          ... & ... & ... & ...& ... \\
          0 & 0 & 0 & ... & \frac{\gamma-\frac{\gamma^{2}}{\alpha}}{\sqrt{\alpha-\frac{\gamma^{2}}{\alpha}}} \\
        \end{array}\right),\]
and by repeating it similarly
{\tiny\[K_{n+3}=\left(
        \begin{array}{ccccc}
          0 & 0 & 0 &... & 0 \\
          0 & 0 & \sqrt{\beta} & ... & 0 \\
          0 & 0 & 0 & ... & 0 \\
          ... & ... & ... & ... & ... \\
          0 & 0 & 0 & ... & 0 \\
        \end{array}
      \right),...,K_{2n+1}=\left(
                      \begin{array}{cccc}
                        0 & 0 & ... & 0 \\
                        0 & 0 & ... & 0 \\
                        \sqrt{\beta} & 0 & ... & 0 \\
                        ... & ...& ...&... \\
                        0 & 0 & ... & 0 \\
                      \end{array}
                    \right),K_{2n+2}=\left(
                      \begin{array}{cccc}
                        0 & 0 & ... & 0 \\
                        0 & 0 & ... & 0 \\
                        0 & \sqrt{\beta} & ... & 0 \\
                        ... & ...& ...&... \\
                        0 & 0 & ... & 0 \\
                      \end{array}
                    \right).\]}
By repeating this algorithm, we see that $\beta$'s are not changed, but $\alpha$'s, $\gamma$'s will be changed in the following way:
$a_{m}=a_{m-1}-\frac{b_{m-1}^{2}}{a_{m-1}}$, for $m=1,2,...,n-1$ and $b_{m}=b_{m-1}-\frac{b_{m-1}^{2}}{a_{m-1}}$, for $m=1,2,...,n-2$, where $a_{0}=\alpha$, $b_{0}=\gamma$. Therefore $a_{m}-b_{m}=a_{m-1}-b_{m-1}=...=a_{0}-b_{0}=2p+\frac{1-p}{n}$ and then $b_{m}=b_{m-1}-\frac{b_{m-1}^{2}}{b_{m-1}+2p+\frac{1-p}{n}}$. By rearrangement of the last equation, we have $\frac{1}{b_{m}}-\frac{1}{b_{m-1}}=\frac{1}{2p+\frac{1-p}{n}}$ and finally
\begin{gather*}
a_{m}=(2p+\frac{1-p}{n})(1+\frac{-p}{-pm+2p+\frac{1-p}{n}}),\quad  1\leq m\leq n-1;\\
b_{m}=(2p+\frac{1-p}{n})(\frac{-p}{-pm+2p+\frac{1-p}{n}}),\quad  1\leq m\leq n-2.
\end{gather*}
\end{proof}
\begin{acknowledgement}{}
Advice given by Professor Grigori G.Amosov has been a great help in writing this paper. I would like to offer my special thanks to him.
\end{acknowledgement}


\end{document}